\documentclass{article}

\usepackage{amssymb,amsmath,latexsym,amscd,amsfonts,bbm, enumerate}  
\usepackage{graphicx}  
\usepackage[usenames, pdftex, dvipsnames]{color}

\newtheorem{theorem}{Theorem}[section]

\newtheorem{deff}[theorem]{Definition}

\newtheorem{lem}[theorem]{Lemma}

\newtheorem{corol}[theorem]{Corollary}

 \newcommand{\qedsymb}{\hfill{\rule{2mm}{2mm}}}  
 \newenvironment{proof}[1][]{\begin{trivlist}  
 \item[\hspace{\labelsep}{\bf\noindent Proof#1:\/}] 
 }{\qedsymb\end{trivlist}}

\newcommand{\ignore}[1]{}  

\newcommand{\inote}[1]{}
\newcommand{\znote}[1]{}
\newcommand{\unote}[1]{}

\newcommand{\NP}{\mathsf{NP}}

\newcommand{\norm}[1]{{\| #1 \|}}  
  
\newcommand{\ket}[1]{{ |{#1} \rangle }}  
\newcommand{\bra}[1]{{ \langle {#1} | }}
\newcommand{\braket}[2]{{ \langle {#1} | {#2} \rangle}}

\newcommand{\orderof}[1]{\mathcal{O}(#1)} 
\newcommand{\bigOmega}[1]{\Omega\mathopen{}\left(#1\right)}
\newcommand{\torderof}[1]{\tilde{\mathcal{O}}(#1)} 
\newcommand{\poly}{\mathrm{poly}} 
\newcommand{\EqDef}{\stackrel{\mathrm{def}}{=}}

\newcommand{\Fig}[1]{Fig.~\ref{#1}}

\newcommand{\Ref}[1]{Ref.~\cite{#1}}

\newcommand{\Thm}[1]{Theorem~\ref{#1}}

\newcommand{\Id}{\mathbbm{1}}

\newcommand{\gs}{\Gamma}
\newcommand{\gsp}{\Gamma^\perp}

\newcommand{\ER}[1]{\mathrm{ER}(#1)}

\newcommand{\Hip}{\mathcal{H^\perp}}

\newcommand{\eps}{\epsilon}
\newcommand{\C}{\mathbb{C}}

\newcommand{\Ht}{H^{(t)}}
\begin{document}
\begin{center}
  {\Large An area law and sub-exponential algorithm for 1D systems} 
\end{center}

\begin{center}
  Itai Arad\footnote{The Hebrew University}, 
  Alexei Kitaev\footnote{California Institute of Technology}, 
  Zeph Landau\footnote{UC Berkeley}, 
  Umesh Vazirani\footnote{UC Berkeley}
\end{center}

\begin{abstract}
  We give a new proof for the area law for general 1D gapped
  systems, which exponentially improves Hastings' famous
  result~\cite{ref:Has07}. Specifically, we show that for a chain of
  $d$-dimensional spins, governed by a 1D local Hamiltonian with a
  spectral gap $\eps>0$, the entanglement entropy of the ground
  state with respect to any cut in the chain is upper bounded by
  $\orderof{\frac{\log^3 d}{\eps}}$. Our approach uses the framework
  of~Refs.~\cite{ref:FOCS2011-AreaLaw,ref:DL-AreaLaw-2011} to
  construct a Chebyshev-based AGSP (Approximate Ground Space
  Projection) with favorable factors. However, our construction uses
  the Hamiltonian directly, instead of using the Detectability
  lemma, which allows us to work with general (frustrated)
  Hamiltonians, as well as slightly improving the $1/\eps$
  dependence of the bound in~\Ref{ref:DL-AreaLaw-2011}. To achieve
  that, we establish a new, ``random-walk like'', bound on the
  entanglement rank of an arbitrary power of a 1D Hamiltonian, which
  might be of independent interest: $\ER{H^\ell} \le (\ell
  d)^{\orderof{\sqrt{\ell}}}$. Finally, treating $d$ as a constant,
  our AGSP shows that the ground state is well approximated by a
  matrix product state with a sublinear bond dimension
  $B=e^{\torderof{\log^{3/4}n/\eps^{1/4}}}$. Using this in
  conjunction with known dynamical programing algorithms, yields an
  algorithm for a $1/\poly(n)$ approximation of the ground energy
  with a subexponential running time $T\le
  \exp\big(e^{\torderof{\log^{3/4}n/\eps^{1/4}}}\big)$.

\end{abstract}

\section{Introduction}

Understanding the structure and complexity of ground states of local
Hamiltonians is one of the central problems in Condensed Matter
Physics and Quantum Complexity Theory.  In gapped systems, a
remarkably general conjecture about the structure of ground states,
{\it The Area Law}, bounds the entanglement that such states can
exhibit. Specifically, for any subset $S$ of particles, it bounds
the entanglement entropy of $\rho_S$, the reduced density matrix of
the ground state restricted to $S$, by the surface area of $S$,
i.e., the number of local interactions between $S$ and
$\overline{S}$~\cite{ref:AL-rev}.

Although the general area law remains an open conjecture, a lot of
progress has been made on proving it for 1D systems.  The
breakthrough came with Hastings' result~\cite{ref:Has07}, which
shows that the entanglement entropy across a cut for a 1D system is
a constant independent of $n$, the number of particles in the
system, and scales as $e^{\orderof{\frac{\log d}{\epsilon}}}$, where $d$ is
the dimension of each particle and $\epsilon$ is the spectral gap.
This result implies that the ground state of a gapped 1D Hamiltonian
can be approximated in the complexity class $\NP$.

In this paper, we:
\begin{itemize}

  \item Give an exponential improvement to 
    $\torderof{\frac{\log^3d}{\eps}}$ in the bound of entanglement
    entropy for the general (frustrated) 1D Hamiltonians.  The
    dependence on the gap even improves the previous best bound for
    frustration free 1D Hamiltonians and may possibly be tight to
    within log factors.
  
  \item Prove the existence of sublinear bond dimension Matrix 
    Product State approximations of ground states for general 1D
    Hamiltonians. This implies a subexponential time algorithm for
    finding such states thus providing evidence that this task is
    not $\NP$-hard.
  
\end{itemize}

\noindent 
We also establish the following properties of local Hamiltonians
which may be of independent interest:
\begin{itemize}
  \item ``Random walk like'' behavior of entanglement: for a 1D
    Hamiltonian $H$, the Entanglement Rank (ER) of $H^{\ell}$ is
    bounded by $(\ell d)^{O(\sqrt{\ell})}$.
    
  \item Let $H'$ be the Hamiltonian consisting only of terms acting 
    on a subset $S$ of particles. Then the ground state of $H$ has
    an exponentially small amount of norm in the "high" energy
    spectrum of $H'$: the total norm with energy above $t$ is
    $2^{-\Omega (t - |\partial S|)}$ where $|\partial S|$ is the
    size of the boundary between $S$ and $\overline{S}$.
\end{itemize}

The work here has its origins in the combinatorial approach
of~\cite{ref:DL2}, which used the Detectability lemma, introduced
earlier in~\cite{ref:DL}, to give a very different proof of
Hastings' result for the special case of frustration-free 
Hamiltonians. The results there were greatly strengthened
in~\cite{ref:FOCS2011-AreaLaw} and \cite{ref:DL-AreaLaw-2011}, which
introduced Chebyshev polynomials in conjunction with the
detectability lemma to construct very strong AGSPs (approximate
ground state projectors), leading to an exponential improvement of
Hastings' bound in the frustration-free case to $O((\frac{\log
d}{\epsilon})^3)$. 

The starting point for our results is to consider a more general
situation where the Hamiltonian obeys the 1D constraint only in a
small neighborhood of $s$ particles around the cut in question (see
\Fig{fig:1dline}). The particles to the left and to the right of this small
neighborhood are acted upon by multi-particle Hamiltonians $H_L$ and
$H_R$ respectively. Constructing an AGSP for the new Hamiltonian is
now much simpler, since the Hamiltonian has small norm: the AGSP is
just a suitable Chebyshev polynomial of the Hamiltonian. In the
frustration-free case, the new Hamiltonian has the same ground state
as the original Hamiltonian, and this leads to a much simpler (and
slightly stronger) proof of the Area Law. In the general, 
frustrated case, there is a tradeoff between the norm of the new
Hamiltonian and how close its ground state is to that of the
original Hamiltonian. To establish an area law, we must now consider
a sequence of Hamiltonians whose ground states converge to the
ground state of the original Hamiltonian, and derive an entropy
bound from the tradeoff between the rate of convergence and the rate
of increase of entanglement rank. 

\begin{figure}
  \begin{center}
    \includegraphics[scale=1]{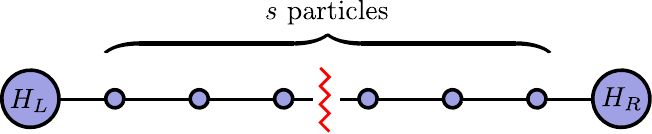}  
  \end{center}
  \caption{\label{fig:1dline}
  The 1D setting. We focus on a segment of $s$ particles
  around the cut, denoting the multiparticle Hamiltonians to the
  left and right of the segment by $H_L$ and $H_R$ respectively.}
\end{figure}

\section{Background:  Approximate Ground State Projectors and their consequences}

The overall strategy is to start with a product state $\ket{\psi}$
and repeatedly apply some operator $K$ such that
$\frac{1}{\norm{K^j\ket{\psi}}}K^j\ket{\psi} $ approximates the
ground state and the entanglement rank of $K^j\ket{\psi}$ is not too
large.  This property of an operator $K $ is captured in the
following definition of an approximate grounds state projection
(AGSP):
\begin{deff}[An Approximate Ground-Space Projection (AGSP)] \ \\
  \label{def:AGSP} 
  Consider a local Hamiltonian system $H=\sum_i H_i$ on a 1D chain,
  together with a cut between particles $i^*$ and $i^*+1$ that
  bi-partitions the system. We say that an operator $K$ is a $(D,
  \Delta)$-Approximate Ground Space Projection (with respect to the
  cut) if the following holds:
  \begin{itemize}
    \item \textbf{Ground space invariance:} for any ground state
      $\ket{\gs}$, $K\ket{\gs} = \ket{\gs}$.
      
    \item \textbf{Shrinking:} for any state $\ket{\gsp}\in \Hip$,
      also $K\ket{\gsp}\in\Hip$, and $\norm{K\ket{\gsp}}^2 \le
      \Delta$.
      
    \item \textbf{Entanglement:} 
      the entanglement rank of $K$, as an element of the tensor
      product of two operator spaces (for the first and the second
      part of the system), is at most $D$.

  \end{itemize}
  (The last condition implies that the operator $K$ changes the
  entanglement rank of an arbitrary quantum state $\ket{\phi}$ at
  most by factor of $D$, i.e.\ $\ER{K \ket{\phi}} \le D
  \cdot\ER{\phi}$.)
\end{deff}
  
The parameters $\Delta$ and $D$ capture the tradeoff between the
rate of movement towards the ground state and the amount of 
entanglement that applying the operator $K$ incurs. 
In~\cite{ref:FOCS2011-AreaLaw,ref:DL-AreaLaw-2011}, it was shown
that a favorable tradeoff gives an area law:
\begin{lem}
  \label{lem:mu1} 
  If there exists an $(D, \Delta)$-AGSP with $D\cdot\Delta \leq
  \frac{1}{2}$, then there is a product state
  $\ket{\phi}$ whose overlap with
  the ground state is $\mu=|\braket{\gs}{\phi}| \ge 1/\sqrt{2D}$.
\end{lem}
  
\begin{lem}
\label{lem:S} 
  If there exists a product state whose overlap with the ground
  state is at least $\mu$, together with a $(D, \Delta)$-AGSP,
  then the entanglement entropy of $\ket{\gs}$ is bounded by
  \begin{align}
  \label{eq:Sb}
    S\le \orderof{\frac{ \log \mu^{-1}}{\log \Delta^{-1}} }\cdot\log D \ .
  \end{align}
 \end{lem}

Combined, the above two lemmas give conditions for an area law:

\begin{corol}[Area Law]
  \label{col:half}
  If there exists an $(D,\Delta)$-AGSP such that
  $D\cdot\Delta\le\frac{1}{2}$, the ground state entropy is bounded 
  by:
  \begin{align}
  \label{eq:AL}
    S\le \orderof{1}\cdot\log D \ .
  \end{align} 
 \end{corol}

\section{Overview}

The results here rely on the construction of a suitable AGSP that
allows the application of Corollary \ref{col:half}.  The first
critical step is to exchange local structure far from the cut for a
valuable {\bf reduction in the norm} of the Hamiltonian.  To do
this, we isolate a neighborhood of $s+1$ particles around the cut in
question, and then separately truncate the sum of the terms to the
left and to the right of these $s+1$ particles.  Specifically, we
define the truncation of an operator as follows:
\begin{deff}[Truncation] 
\label{d:trunc} 
  For any self-adjoint operator $A$, form $A^{\leq t}$, the
  \emph{truncation} of $A$, by keeping the eigenvectors the same,
  keeping the eigenvalues below $\leq t$ the same, and replacing any
  eigenvalue $\geq t$ with $t$.
\end{deff}
We then define $\Ht= (\sum_{i<1} H_i)^{\leq t} + H_1 + \dots + H_{s}
+ (\sum_{i>s}H_i)^{\leq t}$, where the $s$ middle terms act on the 
the isolated string of $s+1$ particles around the cut.  The result
is a Hamiltonian $H$ that is now norm bounded by $u=s+2t$ acting on
$n$ particles with the following structure: 
\begin{equation} 
\label{e:ham}
  H = \Ht= H_L + H_1 + H_2 + \dots + H_{s} + H_R,
\end{equation}
where each $H_i $ are norm bounded by $1$ and acts locally on
particles $m+i$ and $m+i+1$, $H_L$ acts on particles $1, \ldots , m$
and $H_R$ acts on particles $m+s+1, \ldots , n$.  We are interested
in the entanglement entropy across the cut in the middle, i.e.,
between particles $m+s/2$ and $m+s/2+1$. In the frustration free
case, it is clear that the ground state of $\Ht$ is the same as that
of the original Hamiltonian and it can be shown that the spectral
gap is preserved for some constant value of $t$.  For the frustrated
case, the ground state of $\Ht$ is no longer that of the original
Hamiltonian and a limiting argument (see below) will be needed to
complete the proof.

Having reduced the problem to a Hamiltonian with bounded norm $u$ of
the form (\ref{e:ham}), we turn to the next critical step of
constructing the AGSP, the {\bf use of Chebyshev polynomials} to
approximate the projection onto the ground state. We begin with a
suitably modified Chebyshev polynomial $C_{\ell}(x)$ of degree
$\ell$ with the properties that $C_{\ell}(0)=1$ and $|C_{\ell}(x)|
\leq e^{-\bigOmega{\ell\sqrt{\eps/u}}}$ for $\eps \leq x \leq u$. 
The AGSP is then $K= C_{\ell} (H)$ and it is clear that
$\Delta=e^{-\bigOmega{\ell\sqrt{\eps/u}}}$.

Bounding the ER for $K$ requires important new ideas. We may take
the approach of~\cite{ref:FOCS2011-AreaLaw,ref:DL-AreaLaw-2011} as a
starting point and expand $H^{\ell}$ into terms of the form
$H_{j_1}\cdots H_{j_\ell}$. For each such term, there is some $i$
such that $H_i$ occurs at most $\ell/s$ times. Thus, the
entanglement rank of the given term across cut $i$ is less or equal
to $d^{2\ell/s}$. The ER across the middle cut is at most times
$d^{s}$ times greater, which gives an upper bound $d^{2\ell/s+s}$.
The difficulty is that the number of terms, $(s+2)^{\ell}$, is too
large. To address this issue, we introduce formal commuting
variables $Z_i$ and consider the polynomial \[P(Z)=(H_L Z_0 + H_1
Z_1 + \dots + H_R Z_{s+1}) ^{\ell}= \sum_{a_0 +\cdots +
a_{s+1}=\ell} f_{a_0,\ldots,a_{s+1}}Z_0 ^{a_0} Z_1 ^{a_1} \dots
Z_{s+1}^{a_{s+1}}.\] In particular,
$H^{\ell}=\sum_{a_0,\ldots,a_{s+1}}f_{a_0,\ldots,a_{s+1}}$. This
expression has fewer terms, namely, $\binom{\ell+s+1}{s+1}$. As
before, for each multi-index $(a_0,\ldots,a_{s+1})$ there is some
$i$ such that $a_i\le l/s$. If $i$ is fixed, a linear combination of
the corresponding operators $f_{a_0,\ldots,a_{s+1}}$ can be
generated as follows. We restrict our attention to only those terms
in $P(Z)$ where $Z_i$ appears at most $\ell/s$ times and assign
arbitrary values to the variables $Z_0,\ldots,Z_{s+1}$. The ER of
the resulting operator is estimated using the representation $P(Z)
=(A +H_i Z_i +B)^{\ell}$, where $A$ and $B$ commute. We then use a
{\bf polynomial interpolation} argument to express each
$f_{a_0,\ldots,a_{s+1}}$, their sum $H^{\ell}$, and finally, the
operator $K$. Thus we prove that the ER of $K$ is at most
$D=(dl)^{\orderof{l/s+s}}$.

Applying Theorem \ref{col:half} to the above AGSP with $\ell=
O(s^2)$, $s= \tilde{O}(\log^2(d)/ \eps)$ yields our Area Law for
frustration free Hamiltonians, providing an entanglement entropy
bound of the form $\tilde{O}( \log ^3 (d)/\eps)$.

To address the frustrated case, a third critical result is needed: 
that the ground states of $\Ht$ are very good approximations of the
ground state of the original Hamiltonian.  Intuitively, the
structure of the small eigenvectors and eigenvalues of $\Ht$ should
approach those of $H$ as $t$ grows and we show that to be the case,
showing a {\bf robustness theorem}: that the ground states of $\Ht$
and $H$ are exponentially close in $t$ and the spectral gaps are of
the same order.

We would like to apply Theorem \ref{col:half} to an AGSP for $\Ht$,
for $t$ sufficiently large, however, if we try to do this in one
step, the ER cost becomes a large function of $t$.  Instead we use a
well chosen arithmetic sequence $t_0, t_1, \dots$ and the associated
AGSP's to $H ^{(t_i)}$ to guide the movement towards the ground
state.  The robustness theorem allows for very rapid convergence,
the result of which is the area law in the general (frustrated)
case.

\section{Approximate Ground State Projector} 
\label{s:2}

Consider a Hamiltonian $H$ acting on $n$ particles with the
following structure: $H = H_L + H_1 + H_2 + \dots + H_{s} + H_R$,
where $H_i$ acts locally on particles $m+i$ and $m+i+1$, $H_L$ acts
on particles $1, \ldots , m$ and $H_R$ acts on particles $m+s+1,
\ldots , n$. Assume that $H$ has a unique ground state
$\ket{\Gamma}$ with energy $\eps_0$ and that the other eigenvalues
belong to the interval $[\eps_1,u]$.  Let $\eps = \eps_1 - \eps_0$
denote the spectral gap. We wish to bound the entanglement entropy
of $\ket{\Gamma}$ across the middle cut, $i=s/2$. (In our notation,
cut $i$ separates the particles $m+i$ and $m+i+1$.)

We define the AGSP as $K=C_{\ell}(H)$, where $C_{\ell}$ is a
polynomial that satisfies the conditions below for a suitable value
of $\Delta$.
\begin{enumerate}
  \item $C_{\ell}(\eps_0)=1$;
  \item $|C_{\ell}(x)|\le\sqrt{\Delta}$ for $\eps_1 \leq x \leq u$.
\end{enumerate}
It follows that $K\ket{\Gamma}=\ket{\Gamma}$ and that the
restriction of $K$ to the orthogonal complement of $\ket{\Gamma}$
has norm less or equal to $\sqrt{\Delta}$.

\begin{lem}
  There exists a degree $\ell$ polynomial $C_{\ell}$ that satisfies
  the above conditions for
  \begin{align*}
    \sqrt{\Delta}=2\,e^{-2\ell \sqrt{(\eps_1-\eps_0)/(u-\eps_0)}}.
  \end{align*}
\end{lem}

\begin{proof}
  We construct $C_{\ell}$ by a linear rescaling of the Chebyshev
  polynomial $T_{\ell}$, which is defined by the equation
  $T_{\ell}(\cos\theta)=\cos(\ell\theta)$. It follows immediately
  that $|T_{\ell}(x)|\le 1$ for $x\in[-1,1]$. If $x>1$, the equation
  $\cos\theta=x$ has a complex solution, $\theta=it$, where $\cosh
  t=x$. In this case, $T_{\ell}(x)=\cosh(\ell
  t)\ge\frac{1}{2}e^{\ell t}$. Since $t\ge 2\tanh(t/2)=
  2\sqrt{(x-1)/(x+1)}$, we conclude that
  \begin{align*}
    T_{\ell}(x)\ge \frac{1}{2}e^{2\ell\sqrt{(x-1)/(x+1)}}.
  \end{align*}
  Now, let
  \begin{align*}
    C_{\ell}(y)=\frac{T_{\ell}(f(y))}{T_{\ell}(f(\eps_0))},\qquad
    \text{where}\quad f(y)=\frac{u+\eps_1-2y}{u-\eps_1}.
  \end{align*}
  The function $f$ maps $\eps_1$ to $1$ and $u$ to $-1$, hence
  $|C_{\ell}(y)| \le\frac{1}{T_{\ell}(f(\eps_0))}$ for
  $y\in[\eps_1,u]$. The bound for $T_{\ell}(x)$ with $x=f(\eps_0)$
  matches the expression for $\Delta$ because
  $\frac{\eps_1-\eps_0}{u-\eps_0}
  =\frac{f(\eps_1)-f(\eps_0)}{f(u)-f(\eps_0)}=\frac{x-1}{x+1}$.
\end{proof}

\begin{lem} 
\label{l:agsp}
  The entanglement rank of $K=C_{\ell}(H)$ (where $C_{\ell}$ is an
  arbitrary degree $\ell$ polynomial) is bounded by
  $D=(d\ell)^{\orderof{\max\{\ell/s, \sqrt{\ell}\}}}$.
\end{lem}

\begin{proof}
  W.l.o.g.\ we may assume that $s\le\sqrt{\ell}$. If that is not the
  case, we can reduce $s$ to $\sqrt{\ell}$ by joining some of the
  $H_{j}$'s with either $H_{L}$ or $H_{R}$. This does not change the
  actual entanglement rank or the required bound. After this
  reduction, the bound can be written as
  $(d\ell)^{\orderof{\ell/s}}$.

  $K= C_{\ell}(H)$ is a linear combination of $\ell+1$ powers of
  $H$, and we will bound the entanglement rank added by each,
  focusing on the worst case $H^{\ell}$. Let us first consider the
  expansion $H^{\ell}=\sum_{j_1,\dots,j_\ell}H_{j_1}\dots
  H_{j_\ell}$. It has too many terms to be useful, but we can group
  them by the number of occurrences of each $H_{j}$. To this end, we
  introduce a generating function, which is a polynomial in formal
  commuting variables $Z_0,\dots,Z_{s+1}$:
  \begin{align*}
    P_{\ell}(Z)=(H_L Z_0 + H_1 Z_1 + \dots + H_R Z_{s+1})^{\ell}
    = \sum_{a_0+\cdots+a_{s+1}=\ell} f_{a_0,\ldots,a_{s+1}}
      Z_0^{a_0} Z_1^{a_1} \dots Z_{s+1}^{a_{s+1}}.
  \end{align*}
  Each coefficient $f_{a_0,\ldots,a_{s+1}}$ is the sum of products
  $H_{j_1}\dots H_{j_\ell}$, where each $H_{j}$ occurs exactly $a_j$
  times. We are interested in estimating the ER of
  $H^{\ell}=\sum_{a_0,\ldots,a_{s+1}}f_{a_0,\ldots,a_{s+1}}$.

  We start by noticing that for each multi-index
  $(a_0,\ldots,a_{s+1})$, there is some $i\in\{1,\ldots,s\}$ such
  that $a_i\le l/s$. Thus, $H^{\ell}=\sum_{i=1}^{s}
  \sum_{k=0}^{\ell/s} Q_{i,\ell k}$, where $Q_{i,\ell k}$ includes
  some of the operators $f_{a_0,\ldots,a_{s+1}}$ such that $a_i=k$
  and $\sum_{j\not=i}a_j=\ell-k$. (This decomposition of $H^{\ell}$
  is not unique.) We will, eventually, bound the ER of each
  $Q_{i,\ell k}$. To do that, we first define a generating function
  that includes all the matching $f_{a_0,\ldots,a_{s+1}}$'s:
  \begin{align*}
    P_{i,\ell k}(Z)
    =\sum_{\substack{a_i=k\\ \sum_{j\not=i}a_j=\ell-k}} f_{a_0,\ldots,a_{s+1}}
    \prod_{j\not=i}Z_{j}^{a_j}.
  \end{align*}
  This sum has $t=\binom{\ell-k+s}{s}$ terms. The variable $Z_i$ is
  excluded, or we may consider it equal to $1$. When the remaining
  variables are assigned definite values, $Z\in\C^{s+1}$, we obtain
  a linear combination of the operators $f_{a_0,\ldots,a_{s+1}}$.
  The key observation is that such linear combinations have full
  rank, i.e.\ there are $t$ distinct values of $Z\in\C^{s+1}$ such
  that the corresponding $\{P_{i,\ell k}(Z)\}$ form a basis in the
  space of operators of the form $\sum
  c_{a_0,\ldots,a_{s+1}}f_{a_0,\ldots,a_{s+1}}$, where
  $c_{a_0,\ldots,a_{s+1}}\in\C$ and the sum runs over the support of
  $P_{i,\ell k}$. In particular, $Q_{i,\ell k}$ is a linear
  combination of $t$ operators of the form $P_{i,\ell k}(Z)$.

  For a fixed $Z$, the operator $P_{i,\ell k}(Z)$ can be obtained as
  follows. We write $P_{\ell}(Z)=(A+H_i+B)^{\ell}$, where
  $A=\sum_{j<i}H_jZ_j$ and $B=\sum_{j>i}H_jZ_j$, and then collect
  the terms with $H_i$ appearing exactly $k$ times. Since $A$ and
  $B$ commute, such terms have the form $A^{a_0}B^{b_0} H_i \cdots
  H_iA^{a_k}B^{b_k}$. There are $\binom{\ell+k}{2k+1}$ distinct
  terms like that, and the ER of each term across cut $i$ is at most
  $d^{2k}$. The ER across the middle cut is bounded by that number
  times $(d^2)^{|i-s/2|}\le d^s$. Combining all factors, we find
  that
  \begin{align*}
    \ER{Q_{i,\ell k}}\le \binom{\ell-k+s}{s}\binom{\ell+k}{2k+1}\,d^{2k+s}
      \le \ell^{\orderof{s}}\ell^{\orderof{\ell/s}}d^{2\ell/s+s}\le
        (d\ell)^{\orderof{\ell/s}}.
  \end{align*}
  Here we have used the fact that $k\le\ell/s$ and $s\le\sqrt{\ell}$.
  The summation over $i$ and $k$ does not change this asymptotic form.
\end{proof}

Lemma~\ref{l:agsp} gives a non-trivial tradeoff between the
entanglement rank $D$ and shrinking coefficient $\Delta$ of the
operator $K$.  By suitable choice of parameters this will give the
desired $(D,\Delta)$-AGSP such that $D\cdot\Delta\le\frac{1}{2}$ and
Corollary \ref{col:half} will apply. One issue that we will have to
address is the bound $t$ on the norms of $H_L$ and $H_R$.  We first
tackle the case of frustration free Hamiltonians, where we can
assume W.L.O.G. that $t = O(1)= |H_L| = |H_R|$:

Let $H' = \sum H_i$ be a frustration free Hamiltonian with spectral
gap $\epsilon$.  For $t$ chosen in a moment, define $H_L=(\sum_{i
\leq m} H_i)^{\leq t}$ and $H_R=(\sum_{i \geq m+ s+1}H_i)^{\leq t}$
to be the truncation of the Hamiltonian acting on the left and right
ends of the line.  Set $H= H_L + H_1 + H_2 + \dots + H_s + H_R$ so
it is in the form as above.  Clearly $H$ has the same ground state
as $H'$.  Since $\eps_0 =0$, Lemma \ref{l:gap} (below) yields that
for $t=O(\frac{1}{\eps})$ the Hamiltonian $H$ has a gap that is at
least a constant times $\epsilon$.

\begin{theorem}  
\label{t:ent} 
  For a frustration free Hamiltonian $H' = \sum H_i$
  with gap $\eps$ the entanglement entropy is $O(\frac{\log^3
  d}{\epsilon})$.
\end{theorem}

\begin{proof} 
  From Lemma \ref{l:gap}, for $t= O(\frac{1}{\eps})$, we have that
  $H$ has the same ground state and gap of the same order as $H'$.
  Recall lemma \ref{l:agsp} applied to $H$ describe an AGSP with
  bounds $\Delta = e^{-\frac{\ell \sqrt{\eps}}{\sqrt{s+2t}}}$ and $D
  =(\ell +1)\binom{\ell + {s}}{{s}}^2 (\frac{\ell}{s} +1)\binom{\ell
  + \frac{\ell}{s}}{2\frac{\ell}{s}} d^{2\ell/ s} d^s$. Set $\ell =
  s^2/2$, $c= \sqrt{s}/\sqrt{s+2t}$ so that $\Delta =
  e^{-c\epsilon^{1/2}s^{3/2}}$, and $D \leq \binom{(s^2 + s)/2}{s}
  ^4 d^{2s}$. 

  Write the condition $D \Delta < 1/2$ as $\log D < \log
  \frac{1}{\Delta} - 1$. $\log \frac{1}{\Delta} =
  {c\epsilon^{1/2}s^{3/2}}$, and $\log D = O(s (log d + \log (s^2 +
  s)))$.  Thus we can satisfy the condition with $s = O(\frac{log^2
  d}{\epsilon})$, and therefore $\log D = O((\frac{log^3
  d}{\epsilon}))$ and the result follows directly from Corollary
  \ref{col:half}.

\end{proof}

\section{Low bond dimension MPS for frustration free 1D Hamiltonians}

We can use these results to show the existence of a matrix product
state of sub-linear bond dimension of size $\exp (O(
\eps^{-\frac{1}{3}} \log ^{\frac{2}{3}} n))$, that approximates a
ground state $\ket{\Gamma}$ of a gapped frustration free 1D
Hamiltonian to within $\frac{1}{poly(n)}$.  To show the existence of
a matrix product state of bond dimension $B$ within $\delta$ of the
$\ket{\Gamma}$, it suffices to show the existence of a state of
entanglement rank $B$ within $\frac{\delta}{n}$ of $\ket{\Gamma}$.

We've shown the existence of a state $\ket{\psi}$ with constant
overlap with $\ket{\Gamma}$ and entanglement rank $O(\frac{\log ^3
d}{\eps})$.  To this state, we would like to apply an AGSP with
$\Delta = \frac{1}{poly(n)}$.  Just as in the proof of Theorem
\ref{t:ent}, we choose $\ell =s^2$ and Lemma \ref{l:agsp}
establishes that a $\Delta = e^{-c\epsilon^{-1/2}s^{3/2}}$, $D= exp
(s \log d)$ AGSP exists.  Setting $s=O(\eps ^{-\frac{1}{3}} \log
^{\frac{2}{3} }n)$, we have $\Delta = \frac{1}{poly(n)}$ with $D =
\exp (O( \eps^{-\frac{1}{3}} \log ^{\frac{2}{3}} n))$.

\section{Frustrated Case}

\subsection{The operator $\Ht$}

We consider the ground state $\ket{\Gamma}$ of a local Hamiltonian
$H' = \sum H_i '$, where $H_i'$ acts locally on the particles $i$
and $i+1$, and $0 \leq H_i' \leq \Id$.  We assume the Hamiltonian
$H'$ has a unique ground state with energy $\eps_0$ and next lowest
energy $\eps_1$; let $\eps = \eps_1 - \eps_0$ denote the spectral
gap.  It is easy to see that in such a case $\eps \leq 1$.

Ideally we wish to replace $H'$ with some Hamiltonian $H$ with the
same ground state $\ket{\Gamma}$ and spectral gap $\eps$, but with
smaller norm, so that the AGSP from Section \ref{s:2} yields a good
bound on the entanglement entropy of $\ket{\Gamma}$.  Towards that
goal we consider Hamiltonians of the following more general form:

\begin{equation}
\label{e:form}
 H= H_L + H_1 + H_2 + \dots + H_{s} + H_R, 
\end{equation} 
where $H_i$ acts locally on particles $m+i$ and $m+i+1$, $H_L$ acts
on particles $1, \ldots , m$ and $H_R$ acts on particles $m+s+1,
\ldots , n$.  We further require that $H_L$ and $H_R$ are positive
and $|H_1 + \dots H_{s-2}| \leq s$.

We now show that $\ket{\Gamma}$ is the ground state of such a
Hamiltonian $H$ with the added properties that the ground energy of
$H_L$, $H_R$ and $\sum_{i=4}^{s-3} H_i$ are all $0$ and $0\leq H_i
\leq 1$ for $i=1,2,3, s-2, s-1, s$.  By enforcing these properties,
we have $\eps_0 \leq 6$. We do this by setting:
\begin{itemize}
  \item $H_{L}= \sum_{i=1}^ m H_i' - c \mathbbm 1$ where $c$ is the
    ground energy of $\sum_{i=1}^ m H_i' $,
    
  \item  $H_{R}= \sum_{i=m+s+1}^ n H_i' - c' \mathbbm 1$ where $c'$ is
    the ground energy of $\sum_{i=m+s+1}^ n H_i' $,
    
  \item $H_i= H_{m+i} '$ for the six values $i=1,2,3,s-2,s-1, s$,
  
  \item $H_i = H_{m+i} - \frac{d}{s-7} \mathbbm 1$ for $4\leq i\leq
    s-3$, where $d$ is the ground energy of $\sum_{i=4}^{s-3} H_i'$.
\end{itemize}

It is easily verified that $H$ is of the form (\ref{e:form}), and
since the difference between $H$ and $H'$ is a multiple of the
identity, $H$ has ground state $\ket{\Gamma}$.  The Hamiltonian $H$
has ground energy $\leq 6$ since the tensor product of the ground
states for the disjoint operators $H_L$, $\sum_{i=4}^{s-3} H_i$, and
$H_R$ only can have non-zero energy on $H_1 + H_2 + H_3 + H_{s-2}
+H_{s-1} + H_{s}$.  

To bound the norm (so as to effectively apply the AGSP from Section
\ref{s:2}) we use the previously defined notion of truncation
(Definition \ref{d:trunc}).  If for any self-adjoint operator $A$,
we let $P_t$ be the projection into the subspace of eigenvectors of
$A$ with eigenvalues $\le t$, then 
\begin{align}
  \label{def:HR-trunc} 
  A^{\leq t} =P_t A P_t + t(1-P_t) \ .
\end{align}

Define $\Ht= (H_L + H_1)^{\leq t} + H_2 + \dots + H_{s-1} +(H_{s} +
H^R)^{\leq t}$; it is easy to verify for $t\geq 0$, $\Ht \leq H$ and
$\Ht \leq (2t +s) \mathbbm 1$.  Unfortunately, the truncated
Hamiltonian $\Ht$ no longer has the same ground state
$\ket{\Gamma}$.  Intuitively, the structure of the small
eigenvectors and eigenvalues of $\Ht$ should approach those of $H$
as $t$ grows.  The Robustness Theorem (Theorem \ref{l:gap}), stated
below and proved in Section \ref{s:3}, verifies this intuition,
showing that for $t$ bigger than some constant, the gap of $\Ht$ is
of the same order as the gap of $H$ and the ground states of $\Ht$
and $H$ are exponentially close in $t$: 
\begin{theorem}[Robustness Theorem]
\label{thm:Ht} \label{l:gap}
  Let $\ket{\Gamma}, \epsilon_0, \epsilon_1$ be the ground state, the
  ground energy and the first excited level of $H$, and let
  $\ket{\phi}, \epsilon_0 ', \epsilon_1 '$ be the equivalent
  quantities of $\Ht$. Then for $t\ge
  \orderof{\frac{1}{\eps_1 - \eps_0}(\frac{\eps_0}{\epsilon_1-\epsilon_0}+1)}$, we have
  \begin{enumerate}[a.]
        \item \label{enum:1} $\epsilon_1' - \epsilon_0' \ge
    \orderof{\epsilon_1-\epsilon_0}$
    \item  \label{enum:2}
      $\norm{\ket{\phi}-\ket{\Gamma}}^2 \le 2^{-\orderof{t}}.       $
   
  \end{enumerate}
\end{theorem}

\subsection{General 1D Area Law} 
\label{s:final}

We would like to apply lemma \ref{l:agsp} to $\Ht$, for $t$
sufficiently large, however, if we try to do this in one step, the
entanglement rank cost becomes a large function of $t$. The overall
plan is therefore to use a well chosen sequence $t_0, t_1, \dots$ to
guide the movement towards the ground state. More concretely we
define a sequence of states $\ket{\psi_0}, \ket{\psi_1}, \ldots$
that converge to $\ket{\Gamma}$, while carefully controlling the
tradeoff between increase in entanglement rank and increase in
overlap with $\ket{\Gamma}$. 

Denote by $1 - \mu_i$, the overlap between the ground state
$\ket{\phi ^{t_i} }$ of $H^{(t_i)}$ and $\ket{\Gamma}$.  We use the
AGSP $K = C_{\ell}(H^{(t_i)})$ from lemma \ref{l:agsp} to move from
state $\ket{\psi_{i-1}}$ to $\ket{\psi_{i}}$, where $\ket{\psi_i}$
has overlap at least $1 - \mu_i$ with $\ket{\phi^{t_i} }$. We will
show that the increase in entanglement rank of each move is small
enough to bound the entanglement entropy of the limiting state which
is the ground state of $H$.

We now put all the ingredients together to prove an area law for
general 1D systems:

\begin{theorem}
  For any Hamiltonian of the form $H = H_L + H_1 + \cdots H_s + H_R$
  with a spectral gap of $\epsilon$, the entanglement entropy of the
  ground state across the $(s/2, s/2+1)$ cut is bounded by
  $O(\frac{\log^3 d} {\epsilon})$.
\end{theorem}

We begin with a lemma:
\begin{lem}  
  \label{l:5} There are constants $t_0$ and $c$ and states
  $\ket{\psi_i}$ with entanglement rank $R_i$, $i=0,1,2 \dots$
  satisfying:
\begin{enumerate}
  \item $| \bra{\psi_i} \ket{\Gamma}| \geq 1 - O(2^{-i}), 
    \ i\geq 0$.
  \item $\log R_0 = O(\frac{\log ^3 d}{\eps})$, $\log R_i = 
    \log R_0 + O(\sum_{j=1}^i \ell_j \log d)$, with $\ell_j =
    O(\sqrt{(t_j +s)/\eps})$.
\end{enumerate}
\end{lem}

{\bf Proof of Lemma:}  

We begin by choosing constants $t_0$, $c$ such that $\Omega(t_0 +
ic) \geq i+4$, for the $\Omega (t)$ appearing in Theorem
\ref{l:gap}.  Setting $t_i= t_0 + ic$, we therefore have 
\begin{equation} \label{e:close}
\norm{\ket{\Gamma}   - \ket{\phi ^{t_i}}}^2 \leq 2^{-(i+4)}
\end{equation} for all $i$.  Similar to the frustration free case, since $t_0$ is constant, choosing $\ell =s^2$ and $s= O(\frac{\log ^2 d}{\eps})$ gives $D\Delta \leq \frac{1}{2}$ in Lemma \ref{l:agsp} and thus by Lemma \ref{lem:mu1} there exists a product state $\ket{\psi}$ such that $|\bra{\Gamma}\ket{\psi}| \leq \frac{1}{\sqrt{2D}}$ where $\log D = O((\frac{log^3 d}{\epsilon})).$   Returning to Lemma \ref{l:agsp}, this time with  $\ell$ chosen so that $\Delta =e^{-\frac{\ell \sqrt{\eps}}{\sqrt{s + 2t_0 - \eps_0}}} = O(\frac{1}{\sqrt{2D}})$, we establish that the state $\ket{\psi_0}= \frac{C_{\ell}(H_{\leq t_0} )\ket{\psi}}{\norm{C_{\ell}(H_{\leq t_0} )\ket{\psi}}}$ has the property that   $\norm{\ket{\phi ^{t_0}} - \ket{\psi_0}}^2 \leq  \frac{1}{16}$, while having entanglement rank $R_0$ with $\log R_0 = O(\frac{\log ^3 d}{\eps}).$ 

We now inductively define $\ket{\psi_i}$ from $\ket{\psi_{i-1}}$
and show that $\norm{ \ket{\phi ^{t_i}} - \ket{\psi_i}}^2 \leq 2^{-i
-4}$.  Applying the triangle inequality to the induction hypothesis
$\norm{ \ket{\phi ^{t_{i-1}}} - \ket{\psi_{i-1}}}^2 \leq 2^{-i -3}$
along with the already established proximity of $ \ket{\phi
^{t_{i-1}}}, \ket{\phi ^{t_i}}$ to $\ket{\Gamma}$ of
(\ref{e:close}), yields
\begin{equation} \label{e:almost}
\norm{ \ket{\phi ^{t_{i}}} - \ket{\psi_{i-1}}}^2 \leq 2^{-i-1}.
\end{equation}  
Our goal is, with only a small amount of added entanglement, to move
$\ket{\psi_{i-1}}$ a little bit closer to $ \ket{\phi ^{t_{i}}}$
which we will accomplish by using a well chosen AGSP.  With $\ell_i
= O(\sqrt{\frac{t_i+ s}{\eps}})$, Lemma \ref{l:agsp} establishes the
existence of a $(D, \frac{1}{32})$ AGSP $K$ for $H^{(t_i)}$ with a
loose bound of $\log D\leq O( \ell_i \log d)$; we apply this AGSP
$K$ to move from $\ket{\psi_{i-1}}$ to $\ket{\psi_i}$ by setting
$\ket{\psi_i}= \frac{ K \ket{\psi_{i-1}}}{\norm{K
\ket{\psi_{i-1}}}}$.  The shrinking property of the AGSP along with
(\ref{e:almost}) establishes $\norm{ \ket{\phi ^{t_i}} -
\ket{\psi_i}}^2 \leq 2^{-i -4}$.

All told we have generated states $\ket{\psi_i}$ with entanglement
rank $R_i = R_0 + \sum_{j\leq i} \ell_i \log d$ and $\norm{
\ket{\phi ^{t_i}} - \ket{\psi_i}}^2 \leq 2^{-i -4}$.  Finally, $
|\braket{\psi_i}{\Gamma}| \geq 1- \norm{ \ket{\Gamma} -
\ket{\psi_i}}^2/2 \geq 1 -2^{-i}$ where the last inequality again
used (\ref{e:close}).


{\bf Proof of Theorem:}

The above lemma gives a series of states of bounded entanglement
rank that converge to the ground state $\ket{\Gamma}$.  Thus if $\{
\lambda_i\}$ are the Schmidt coefficients of $\Gamma$, Lemma
\ref{l:5} \[ \sum_{i=1}^{R_i} \lambda_i ^2 \geq |\bra{\psi_i
}\ket{\Gamma}|^2 \geq 1 - O(2^{-i}). \]

The entropy of $\ket{\Gamma}$ is then upper bounded by summing
$O(2^{-i})\log R_i$ (i.e. the maximal entropy contribution of mass
$O(2^{-i})$ spread over $\log R_i$ terms).  We arrive at a bound of
the entanglement entropy of $\ket{\Gamma}$ given by: \[ \sum_i
O(2^{-i}) (\log R_0 + O(\sum_{j=1}^i \ell_j \log d) )= \log R_0
\sum_ i O(2^{-i}) + \sum_i O(2^{-i}) \sum_{j=1}^i\sqrt{\frac{t_0 +
ic + s}{\eps}} \log d \] \[= O( \log R_0) + O( \sqrt{\frac{t_0
+s}{\eps}}) =O(\frac{\log ^3 (d)}{\eps}).\]

\subsection{Proof of Theorem \ref{thm:Ht}} \label{s:3}

Before giving the proof, we use a simple Markov bound to show that
in a gapped situation, a state with low enough energy must be close
to the ground state:

\begin{lem}[Markov] 
\label{l:mark} 
Let $B$ be a self-adjoint operator with lowest two eigenvalues $\eps
_0 < \eps_1$; denote its lowest eigenvector by $\ket{\psi}$.  Given
a vector $\ket{v}$ with low energy, i.e. such that $\bra{v} H
\ket{v} \leq \eps_0 + \delta $, then $\ket{v}$ is close to
$\ket{\psi}$ in the following sense: 
\begin{align*}
  \norm{\ket{\psi} - \ket{v}}^2 \leq \frac{2\delta}{\eps_1 - \eps_0}
  \ .
\end{align*} 
\end{lem}

\begin{proof}
  Write $\ket{v} = a\ket{\psi} + \sqrt{1-a^2} \ket{\psi ^{\perp}}$
  where $\ket{\psi^{\perp}}$ is orthogonal to $\ket{\psi}$.  The
  energy of $\ket{v}$ then satisfies 
  \begin{align*}
    a^2\eps_0 +(1- a^2) \eps_1 \leq \bra{v} H \ket{v} \leq \eps_0 + \delta , 
  \end{align*}
  and thus $(1-a^2) \leq \frac{\delta}{\eps_1 - \eps_0}$. The
  result follows from noting that $\norm{\ket{\psi} - \ket{v}}^2 =
  (1-a)^2 + (1-a^2)= 2-2a\leq 2(1-a^2)$.
\end{proof}

\begin{proof}[\ of \Thm{thm:Ht}]

  Define $A$ to be the sum of the two terms $A= H_2 + H_{s-1}$. 
  Notice that the operators:
  \begin{align*}
    \{(\Ht -A), H^L + H_1, H_s + H^R, H-A: t\geq 0 \}
  \end{align*}
  all commute with each other and therefore all the operators are
  simultaneously diagonalizable, i.e. they have a common collection
  of eigenstates.  We fix $P_t$ to be the projection onto the
  subspace spanned by those eigenstates of $H-A$ with eigenvalues
  less than $t$. This collection of eigenstates clearly have
  eigenvalues less than $t$ for the operators $H^L+H_1$ and $H_s +
  H^R$ and therefore $H^L+ H_1= (H^L +H_1)^{\leq t}$ and $H_s
  +H^R=(H_s + H^R)^{\leq t}$ on the range of $P_t$.  This allows the
  important observation that 
  \begin{equation} 
  \label{e:1}
    \Ht P_t= H P_t.
  \end{equation}  

  \noindent {\bf Proving part \ref{enum:1}}:

  \znote{ $v_t$, $v_h$ vs $v_{-}, v_{+}$ ?}

  The main idea is to consider the normalized projection of the two
  lowest eigenvectors of $\Ht$ onto the lower part of the spectrum
  of $H-A$ using $P_t$.  Under the assumption that $t$ is
  sufficiently large and the gap of $\Ht$ is sufficiently small
  (relative to the gap of $H$), we show contradictory facts about
  these normalized projections: that they are simultaneously far
  apart from each other (because applying $P_t$ did not move either
  very much) and close to $\ket{\Gamma}$ (because they both have
  energy with respect to $H$ that is close to $\eps_0$).  This
  contradiction allows us to conclude that for $t$ sufficiently
  large, the gap of $\Ht$ must be of the order of the gap of $H$. 

  For every normalized state $\ket{v}$, we define $\ket{v_t} = P_t
  \ket{v}$ and $\ket{v_h}= (1-P_t) \ket{v}$.  Note that by
  (\ref{e:1}),
  \begin{equation} 
    \label{e:2}
    \bra{v_t}H\ket{v_t} = \bra{v_t}\Ht\ket{v_t} 
  \end{equation}
  Our main technical tool is the following lemma that connects the
  energy of a state $\ket{v}$ to that of $\ket{v_t}$:

  \begin{lem} \label{l:1} For any  state $\ket{v}$, 
    \begin{enumerate}
      \item $\norm{\ket{v_h}} \leq 
        \sqrt {\frac{ \bra{v } \Ht \ket{v }}{t}}$,
      \item $\bra{v_t} H \ket{v_t} \leq 
        \bra{v} H \ket{v}  +O(\sqrt{\frac{\bra{v} \Ht \ket{v} }{t}})$.
    \end{enumerate}
  \end{lem}

  \begin{proof}
    The first result follows from
    \begin{align*}
      \bra{v} \Ht \ket{v} \geq \bra{v} (\Ht - A)\ket{v} 
      = \bra{v_t}(\Ht - A)\ket{v_t} +  \bra{v_h} (\Ht - A)\ket{v_h}   \geq
      \bra{v_h} (\Ht - A)\ket{v_h}  \geq t \norm{v_h}^2 \ ,
    \end{align*}
    the last inequality by the definition of $P_t$.

    For the second result,
    \begin{align*}
      \bra{v_t} H \ket{v_t}= \bra{v_t} \Ht \ket{v_t} 
      \leq \bra{v} \Ht \ket{v }  + 2|\bra{v_t} \Ht \ket{v_h}|,
    \end{align*}
    the first equality from (\ref{e:1}) and the second inequality
    from writing $\ket{v_t} = \ket{v} - \ket{v_h}$ and expanding. We
    now bound the second term on the right hand side.  Notice that
    $\bra{v_t} \Ht \ket{v_h}= \bra{v_t} \Ht-A \ket{v_h} +\bra{v_t}
    A \ket{v_h}= \bra{v_t} A \ket{v_h}$.  To $|\bra{v_t} A \ket{v
   _h}|^2$ we apply Cauchy-Schwartz to get $|\bra{v_t} A \ket{v
   _h}|^2 \leq |A \ket{v_t}|^2 |\ket{v_h}|^2 \leq 2
    \sqrt{\frac{\bra{v} \Ht \ket{v} }{t}}$.  
\end{proof}

To prove part \ref{enum:1}, assume $\eps_1' - \eps_0 '\leq
\frac{1}{10} (\eps_1 - \eps_0)$ and denote by $\ket{\phi ^1}$ the
eigenvector of $\Ht$ with eigenvalue $\eps_1'$.  Write $\ket{\phi
_t}= P_t \ket{\phi}$, $\ket{\phi^1_t}= P_t \ket{\phi^1}$. 
Lemma~\ref{l:1} establishes 
\begin{align*} 
  \bra{\phi_t} H \ket{\phi_t} 
    \leq \eps_0 + O(\sqrt{\frac{\eps_0}{t}}) \ ,
\end{align*}
\begin{align*} 
  \bra{\phi^1_t} H \ket{\phi^1_t} 
    \leq \eps_0 + \frac{1}{10} (\eps_1 - \eps_0)  
      + O(\sqrt{\frac{\eps_0 
        + \frac{1}{10} (\eps_1 - \eps_0)}{t}}) .
\end{align*}

Setting $\ket{v} = \frac{\ket{\phi_t}}{\norm{\ket{\phi_t}}},$
$\ket{v'} = \frac{\ket{\phi^1_t}}{\norm{\ket{\phi^1_t}}},$ and
using the above in Lemma \ref{l:mark} yields \begin{align*}\norm{
\ket{\Gamma} - \ket{v}}^2 \leq O(1)\frac{\sqrt{\eps_0}}{\eps_1- \eps
_0} \frac{1}{\sqrt{t}},\end{align*} \begin{align*} \norm{\ket{
\Gamma} -\ket{v'}}^2 \leq \frac{1}{10} + O(1)\frac{ \sqrt{\eps_0 +
\frac{1}{10}(\eps_1-\eps_0)}}{\eps_1 -\eps_0}
\frac{1}{\sqrt{t}}.\end{align*} This establishes, for sufficiently
large $t= \orderof{\frac{\eps_0 +
\frac{1}{10}(\eps_1-\eps_0)}{(\epsilon_1-\epsilon_0)^2}}$, that
$\ket{v}$ and $\ket{v'}$ are both near $\ket{\Gamma}$ contradicting
the fact that they are also almost orthogonal. 
\end{proof}

\noindent {\bf Proving part \ref{enum:2}:}

We are interested in showing that the ground states of $H$ and $\Ht$
are very close together.  Clearly, the ground states of the nearby
Hamiltonians $H-A$ and $\Ht-A$ are identical since they only differ
among the eigenvectors with values above $t$ and so the question
becomes how much the addition of $A$ can change things.  This
reduces to how much the operator $A$ mixes the low and high spectral
subspaces of $H-A$ (i.e. how big the off-diagonal contribution of
$A$ is when it is viewed in a basis that diagonalizes $H-A$).  The
core component of the argument will be the Truncation Lemma (Lemma
\ref{l:smalltail}): that the ground state $\ket{\Gamma}$ is
exponentially close to the range of $P_t$ (i.e.  the low spectral
subspace of $H-A$).  We will combine this result with the fact that
$H$ and $\Ht$ are identical on the range of $P_t$ to argue that
ground states for $H$ and $\Ht$ are exponentially close.

The following lemma captures the bounds necessary for proving the
Truncation Lemma.
  \begin{lem} 
  \label{l:combounds} 
    With $H$, $P_t$, $\ket{\Gamma}$, $\eps_0$ as above we have the
    following:
    \begin{enumerate}
      \item $\norm{(1-P_t)\ket{\Gamma}}^2 
        \leq \frac{2 |\bra{\Gamma} (1-P_t)AP_t \ket{\Gamma}|}{t-\eps_0}$,
      \item  For $t\geq u$, $\norm{(1-P_t) H P_u} 
        = \norm{(1-P_t) A P_u} \leq 2 e^{-\frac{t-u}{8}}.$
    \end{enumerate}
  \end{lem}
\begin{proof}

For 1., by definition, 
\begin{align*}
  \eps= \bra{\Gamma}H\ket{\Gamma}
    = \bra{\Gamma}P_tHP_t\ket{\Gamma} 
      + \bra{\Gamma} (1-P_t) H (1- P_t)\ket{\Gamma} 
      + \bra{\Gamma} P_tH(1-P_t)\ket{\Gamma} 
      + \bra{\Gamma}(1-P_t) H P_t\ket{\Gamma}.
\end{align*}

This gives the bound
\begin{align*} 
  \eps \geq \eps \norm{P_t\ket{\Gamma}}^2 + t
  \norm{(1-P_t)\ket{\Gamma}}^2 - 2| \bra{\Gamma}(1-P_t) A
  P_t\ket{\Gamma}|, 
\end{align*} 
where the second term of the right hand side follows from the
inequality $\bra{\Gamma} (1-P_t) H (1-
P_t)\ket{\Gamma}\leq\bra{\Gamma} (1-P_t) (H-A) (1-
P_t)\ket{\Gamma}$, and the replacement of $H$ with $A$ in the third
term follows from $(1-P_t)(H-A)P_t=(1-P_t)P_t(H-A) = 0$.  Writing
$\norm{P_t\ket{\Gamma}}^2= 1 - \norm{(1-P_t)\ket{\Gamma}}^2$ and
rearranging terms yields statement 1.

For statement 2., the first inequality follows simply from writing
$H= (H-A) +A$ and noting that $H-A$ commutes with $P_u$.  We write
$(1-P_t) AP_u= (1-P_t)e^{- r (H-A)}e^{r(H-A)}Ae^{-r(H-A)} e^{r(H-A)}
P_u$, for an $r>0$ to be chosen later, and noting therefore that

\begin{align*} 
  \norm{(1-P_t) AP_u} 
    \le \norm{(1-P_t) e^{-r(H-A)}}
      \cdot\norm{e^{r(H-A)}Ae^{-r(H-A)} }
    \le  e^{-r(t-u)}\norm{e^{r(H-A)}Ae^{-r(H-A)}} \ .
\end{align*}

The Hadamard Lemma gives the expansion
\begin{align} 
  \label{e:hada} 
    e^{r(H-A)}Ae^{-r(H-A)} &= A + r[H-A,A] 
      +\frac{r^2}{2!} [H-A,[H-A,A]]  \\
      &+ \frac{r^3}{3!} [H-A,[H-A,[H-A,A]]]
     + \cdots \nonumber \\
   &= Q_0 + rQ_1 + \frac{r^2}{2!} Q_2 
     + \frac{r^3}{3!}Q_3 + \cdots \ ,
\nonumber
\end{align}
and we turn to bounding the norm of these operators $Q_i$.  

If we expand $H-A$ and $A$ as the sum of its constituent local terms
$H_j$, each $Q_i$ can be written as a sum of $n_{i}$ terms, each a
product of $H_j$'s; we now bound $n_i$.  Notice that $n_0=2$ and
that $Q_{i-1}$ consists of terms, each of which is a product of at
most $i$ $H_j$'s.  For such a product, there are at most $2 i$ terms
in $H-A$ that do not commute with it.  This implies the recursive
bound $n_{i} \leq 4 i n_{i-1}$ and thus $n_{i} \leq 2\cdot 4^{i} i
!$.  Since each of the terms is norm bounded by $1$, we have
$\norm{Q_i} \leq 2 \cdot 4^i i !$.  Plugging this bound into
(\ref{e:hada}) we have $ e^{r(H-A)}Ae^{-r(H-A)}\le 2\sum_i(4r)^i$;
choosing $r=\frac{1}{8}$ gives a bound of $4$ for (\ref{e:hada}) and
establishes statement 2.
\end{proof}

\begin{lem}[Truncation Lemma] 
\label{l:smalltail}
  For $t > 17$, 
  \begin{equation} 
  \label{e:tl}
    \norm{  ( 1- P_t) \ket{\Gamma}} \leq 2^ {-\Omega (t)}.
  \end{equation}
\end{lem}

\begin{proof}
  We show a discrete version of (\ref{e:tl}): that there exist 
  constants $s=16 + \eps_0$ and $d=16$ such that for integers $n\geq
  0$
  \begin{equation} 
  \label{e:discretetail}
    \norm{(1-P_{s+ nd})\ket{\Gamma}} \leq 2^{-n}
  \end{equation}
  The result will then follow since $s + nd \leq t \leq s+ (n+1)d$
  implies 
  \begin{align*} 
    \norm{(1-P_{t})\ket{\Gamma}} \leq \norm{(1-P_{s+
    nd})\ket{\Gamma}} \leq 2^{-n} \leq 2^{-(\frac{t-s}{d}-1)} = 2^{-
    \Omega (t)}. 
  \end{align*}

  To prove (\ref{e:discretetail}), we will proceed by induction. 
  Clearly the initial case of $n=0$ holds.  Assume that
  $\norm{(1-P_{s+nd})\ket{\Gamma}} \leq 2^{-n},$ for $n<n_0$.
  Define $P_{[0]}=P_s$ and $P_{[j]}= P_{s+jd} - P_{s+ (j-1)d}$ for
  $1\leq j \leq n_0$; thus $P_{s+n_0d}= \sum_{j=0}^{n_0} P_{[j]}$
  and the induction hypothesis implies
  \begin{equation} 
    \label{e:easy}
    \norm{P_{[j]}\ket{\Gamma}}\leq 2^{-j+1}.
  \end{equation}
  for $j< n_0$. By Lemma \ref{l:combounds}, 
  \begin{align*}
   \norm{(1-P_{s+ n_0 d})\ket{\Gamma}}^2\leq \frac{2 | \bra{\Gamma}
   (1-P_{s+n_0d})A P_{s+n_0d} \ket{\Gamma}|} {s+ n_0d-\eps_0}. 
  \end{align*}
  Our goal is to bound the numerator of the right hand side by
  $16\cdot 2^{-2n_0}$; (\ref{e:discretetail}) then follows since the
  denominator is at least $16$.  Write
  \begin{align*} 
    | \bra{\Gamma} (1-P_{s+n_0d})A P_{s+n_0d} \ket{\Gamma}|=
    |\sum_{j=0}^{n_0} \bra{\Gamma} (1-P_{s+n_0d})A
    P_{[j]}\ket{\Gamma}| 
  \end{align*}
  \begin{align*} 
    \leq \sum_{j=0}^{n_0} \norm{ (1-P_{s+n_0d})A P_{[j]}} \norm{(1-
    P_{s +n_0d}) \ket{\Gamma} }\norm{P_{[j]}\ket{\Gamma}} \leq 
    \sum_{j=0}^{n_0} \norm{ (1-P_{s+n_0d})A P_{[j]}} 2^{-j - n_0
    +2}, 
  \end{align*} 
  where the last equation used (\ref{e:easy}) and the fact that
  $\norm{(1- P_{s +n_0d}) \ket{\Gamma} } \leq \norm{(1- P_{s
  +(n_0-1)d}) \ket{\Gamma} }\leq 2^{-(n_0 -1)}.$ Applying Lemma
  \ref{l:combounds} to bound the first term in the sum on the right
  hand side yields:
  \begin{align*} 
    |\bra{\Gamma} (1-P_{s+n_0d})A P_{s+n_0d} \ket{\Gamma}| \leq
    \sum_{j=0}^{n_0} e^{-\frac{(n_0-j)d}{8}} 2^{-(n_0 +j-3)} =
    2^{-2n_0} (8)\sum_{j=0}^{n_0} 2^{ (n_0 -j) (1- d \frac{1}{8\ln
    2})} .
  \end{align*}
  With the choice of $d=16\ln2$, the sum in the
  last term on the right hand side is a geometric series that is
  bounded by $2$ which yields the desired bound of $16 \cdot
  2^{-2n_0}$ and completes the proof of (\ref{e:discretetail}).
\end{proof}

We now use the Truncation Lemma to show that the the projected state
$P_t\ket{\Gamma}$ is exponentially close to eigenvectors of both $H$
and $\Ht$.
\begin{lem}
  The state $\ket{\Gamma_t}
  =\frac{1}{\norm{P_t\ket{\Gamma}}}P_t\ket{\Gamma}$ is an
  approximate eigenvector of both $H$ and $\Ht$ in the following
  sense:
  \begin{equation} 
  \label{e:approxev} \norm{H\ket{\Gamma_t} - \eps_0
    \ket{\Gamma_t}}= \norm{\Ht\ket{\Gamma_t} - \eps_0
    \ket{\Gamma_t}} \leq 2^{-\Omega (t)}.
  \end{equation}
\end{lem}

\begin{proof}
  We begin by writing $H(1-P_t) \ket{\Gamma}= \sum_{i=0}^{\infty} H
  ( P_{t+i +1}-P_{t+i} )\ket{\Gamma}$ and thus
  \begin{align*}
    \norm{H(1-P_t) \ket{\Gamma}} 
      \le \sum_{i=0}^{\infty}
        \norm{H(P_{t+i+1} - P_{t+i})}\cdot\norm{P_{t+i} \ket{\Gamma}} 
      \le \sum_{i=0}^{\infty} (t+(i+1) + 2) 2^{-\Omega(t+i)} 
      \le 2^{-\Omega{(t)}} \ ,
  \end{align*}
  where the bound on the first term follows from the fact that
  $H-A\le(t+i+1)\cdot \Id$ on the range of $P_{t+i+1}$ and
  the bound on the second term from the Truncation Lemma.  

  We write $\eps_0 \ket{\Gamma}= H \ket{\Gamma}= H(1-P_t)
  \ket{\Gamma} + HP_t \ket{\Gamma}$.  We have bounded the first term
  on the right hand side by $2^{-\Omega (t)}$ and it follows simply
  that 
  \begin{equation} 
  \label{e:10} \norm{HP_t\ket{\Gamma} - \eps_0 \ket{\Gamma}} =
    \norm{\Ht P_t\ket{\Gamma} - \eps_0 \ket{\Gamma}} \leq 2 ^{-
    \Omega (t)}, 
  \end{equation}
   where the first equality is from (\ref{e:1}).  Multiplying
   (\ref{e:10}) by $\frac{1}{\norm{P_t \ket{\Gamma}} }\leq 
   \frac{1}{\sqrt{1- 2^{-\Omega (t)}}}$, a constant close to one
   yields (\ref{e:approxev}) since the constant can be absorbed into
   $2^{-\Omega (t)}$.
\end{proof}

We outline the remainder of the argument.  The approximate
eigenvalue property of (\ref{e:approxev}) can be used to show that
$\ket{\Gamma_t}$ is close to an eigenvector of $\Ht$ with
eigenvalue in the range $[\eps_0 - 2^{-\Omega{(t)}}, \eps_0
+2^{-\Omega{(t)}}]$.  By combining a lower bound for the ground
energy of $\Ht$ with the fact (part a.) that $\Ht$ has a gap of
reasonable size, we are able to show that the eigenvector of $\Ht$
near $\ket{\Gamma_t}$ is the ground state of $\Ht$.  The proximity
of $\ket{\Gamma_t}$ to the ground states of both $\Ht$ and $H$ then
establishes the result.

We begin by setting $\delta \EqDef 2^{-\Omega (t)}$.  The approximate
eigenvalue property of (\ref{e:approxev}) implies that there is an
eigenvalue of $\Ht$ within $\delta$ of $\eps_0$, for if not, writing
$\ket{\Gamma_t} = \sum_i c_i \ket{v_i}$ where $\ket{v_i}$ are the
eigenvectors of $\Ht$ with eigenvalues $\lambda_i$, $\norm{\Ht
\ket{\Gamma_t} - \eps_0 \Omega_t}^2 = \sum_i (\lambda_i - \eps_0)
c_i^2 > \delta^2 \sum_i c_i^2= \delta^2$ which contradicts
$(\ref{e:approxev}) $.  We now show that there is only one
eigenvalue in this range and that it is in fact the ground energy
for $\Ht$.  This will follow by lower bounding the energy of $\Ht$
by $\eps_0 + \delta - (\eps_1 ' - \eps_0 ')$.

Decompose an arbitrary state $\ket{\psi}= P_t \ket{\psi} + (1-P_t)
\ket{\psi} = \ket{\psi_1} + \ket{\psi_2}$.  Then 
\begin{align*} 
  \bra{\psi}\Ht \ket{\psi} = \bra{\psi}(\Ht -A)\ket{\psi} +
  \bra{\psi}A \ket{\psi}= \bra{\psi_1}(\Ht -A)\ket{\psi_1} +
  \bra{\psi_2}(\Ht -A)\ket{\psi_2}+ \bra{\psi}A \ket{\psi} \geq t
  \norm{\psi_2}^2 . 
\end{align*} 
A $\ket{\psi}$ for which $t \norm{\psi_2}^2 \geq \frac{\eps_0}{t}$
will therefore have energy at least $\eps_0$.  In the remaining case
of $\norm{\ket{\psi_1}}^2 \geq 1- \frac{\eps_0}{t}$ we bound the
energy of $\ket{\psi}$ as follows:
\begin{equation} \label{e:10b}
  \bra{\psi}\Ht \ket{\psi} \geq \bra{\psi_1}\Ht\ket{\psi_1} + 
  \bra{\psi_2}\Ht \ket{\psi_2} -2|\bra{\psi_1}A \ket{\psi_2}| \geq
  \bra{\psi_1}\Ht\ket{\psi_1} + \bra{\psi_2}\Ht \ket{\psi_2} -
  4\sqrt{\frac{\eps_0}{t}}. 
\end{equation}
Since $\ket{\psi_1}$ is in the range of $P_t$, (\ref{e:1}) implies
$\bra{\psi_1}\Ht\ket{\psi_1}=\bra{\psi_1}H\ket{\psi_1} \geq \eps
_0\norm{\psi}^2$.  Combining this bound with (\ref{e:10b}) gives
$\bra{\psi}\Ht \ket{\psi} \geq \eps_0 - O(\frac{\eps_0}{t})$ and
thus we've shown $\eps_0' \geq \eps_0 - O(\frac{\eps_0}{t})$. Since
(from part a.) $\eps_1' - \eps_0' \geq O( \eps_1 - \eps_0)$, a
choice of $t = O(\frac{\eps_0}{(\eps_1 - \eps_0)^2})$ ensures that
the the ground energy of $\Ht$ is at least $\eps_0 + \delta -
(\eps_1'- \eps_0')$.
  
We've established that both ground energy and the energy of
$\ket{\Gamma_t}$ with respect to $\Ht$ is in the interval $[\eps_0 -
\delta, \eps_0 + \delta]$.  Applying Lemma~\ref{l:mark} gives 
\begin{align*}
  \norm{ \ket{\psi}  - \ket{\Gamma_t}} \leq \delta ,
\end{align*} 
and b. then follows from recalling that $\norm{\ket{\Gamma} -
\ket{\Gamma_t} } \leq \delta$ as well.

\section{Sub exponential algorithm for finding the ground energy of
gapped 1D Hamiltonians}

We now show that the ground state can be well-approximated by a MPS
with a sublinear bond dimension, and, consequently, a
$1/\poly(n)$ approximation of its ground energy can be found in a
subexponential time. To simplify the discussion, we shall treat
$d$ as a constant.

As discussed in the frustration free case, to show that
$\ket{\Gamma}$ can be approximated to within $\frac{1}{poly(n)}$
with an MPS of bond dimension $B=\tilde{O}( \exp (\log
^{\frac{3}{4}} n/\eps^{\frac{1}{4}}))$, it suffices to show for each
cut $(i,i+1)$ the existence of a state with entanglement rank $B$
across that cut that is within $\frac{1}{poly(n)}$ of
$\ket{\Gamma}$.  Theorem \ref{thm:Ht} yields that for $t= O(\log
n)$, $\norm{\ket{\phi^{(t)}} -\Gamma}|| \leq \frac{1}{poly(n)}$, and
therefore we turn to finding a state with entanglement rank $B$
across the cut $(i,i+1)$ that approximates $\ket{\phi^{(t)}}$. 
Applying the AGSP $K= C_{\ell} (\Ht)$ of Lemma \ref{l:agsp}, we have
$\Delta= \frac{1}{poly(n)}$ for $\ell = O(\log n \sqrt{\frac{s +
\log n}{\eps}}))$, and $D=\tilde{O}(\exp(\ell/s +s ))$.  The optimal
choice of $s= \log ^{\frac{3}{4}} n/ \eps^{\frac{1}{4}}$ gives the
desired $B=e^{\torderof{\log ^{\frac{3}{4}}n/\eps^{1/4}}}$.

We can now use this result to bound the complexity of actually
finding the ground energy. For simplicity, we treat $d,\eps$ as
constants. Using recent dynamical programing results
\cite{ref:dprog1,ref:dprog2}, we infer that there exists an
algorithm that runs in time $T = (dBn)^{\orderof{B^2}}\le
\exp(e^{\torderof{\log^{3/4}n/\eps^{1/4}}})$ and finds a
$1/\poly(n)$ approximation of the ground energy. Since
$e^{\torderof{\log^{3/4}n/\eps^{1/4}}}$ is smaller than any finite root of $n$,
it follows that
\begin{corol}
\label{cor:notNP}
  Finding a $1/\poly(n)$ approximation to the ground energy of a
  1D, nearest-neighbors Hamiltonian with a constant spectral gap is
  not $\NP$-hard, unless 3-SAT can be solved in a sub-exponential time.
\end{corol}

\section{Acknowledgments}
\label{sec:Acknowledgements}

We are grateful to Dorit Aharonov, Fernando Brandao, and Matt
Hastings for inspiring discussions about the above and related
topics.

\bibliographystyle{ieeetr}

{~}

\bibliography{QC}

\end{document}